\documentclass[prl,nobalancelastpage,twocolumn,superscriptaddress,nolongbibliography]{revtex4-2}
\pdfoutput=1

\usepackage{color,amsthm,amsmath,amsxtra,amsfonts,dsfont,graphicx,bm,amssymb}
\usepackage[colorlinks=true,linkcolor=blue, citecolor=blue, urlcolor=blue, bookmarks]{hyperref}
\usepackage{centernot}
\usepackage[dvipsnames]{xcolor}
\usepackage{graphicx}
\usepackage{tikz}
\usepackage{braket}
\usepackage{natbib}
\usepackage{url}
\usepackage{hyperref}

\usepackage{fancyhdr}

\newcommand{\be}{\begin{equation}}
\newcommand{\ee}{\end{equation}}
\newcommand{\bea}{\begin{eqnarray}}
\newcommand{\eea}{\end{eqnarray}}
\newcommand{\bse}{\begin{subequations}}
\newcommand{\ese}{\end{subequations}}

\theoremstyle{plain}

\newtheorem{lem}{Lemma}[section]

\newcommand{\prlsection}[1]{{\em {#1}.---~}}

\begin{document}
\title{Super-Extensive Charging Power in the Absence of Global Operations}

\author{Anupam$^\S$}
\email{asarkar@imsc.res.in}

\author{Sheryl Mathew$^\S$}
\email{smathew@imsc.res.in}

\author{Sibasish Ghosh}
\email{sibasish@imsc.res.in}

\thanks{$^\S$These authors contributed equally to this work.}

\affiliation{Optics and Quantum Information Group, The Institute of Mathematical Sciences,
CIT Campus, Taramani, Chennai 600113, India.}
\affiliation{Homi Bhabha National Institute, Training School Complex, Anushakti Nagar, Mumbai 400094, India.}

\begin{abstract}
Quantum batteries have emerged as a platform for investigating whether quantum effects can accelerate energy storage beyond classical limits. Although a variety of charging schemes have reported signatures of quantum advantage, the fundamental physical requirements for achieving superextensive charging power remain insufficiently understood. Here, we show that, in addition to Hamiltonian locality, a key structural property, \emph{$g$-extensiveness}, quantifying the distribution of interaction energy across lattice sites places a fundamental bound on charging performance in spin-lattice models. We prove that superextensive power scaling is possible only when the interaction-energy distribution becomes increasingly nonuniform, with the maximal local weight growing with system size. This criterion explains why many previously studied protocols fail to exhibit superextensive power, even when the Hamiltonians involve large participation numbers. We further demonstrate that this condition is realized in an experimentally relevant interacting model, where, despite fixed interaction order, the charging power scales superextensively. Our results establish $g$-extensiveness as a necessary resource for quantum advantage in direct-charging protocols and provide a systematic framework for identifying and engineering physically feasible quantum batteries capable of outperforming classical counterparts in charging power.
\end{abstract}

\maketitle


\prlsection{Introduction} Since the advent of the concept of quantum battery \cite{Campaioli2024Colloquium, Alicki2013EntanglementBoost, Quach2023QuantumBatteries, Binder2015Quantacell, Bhattacharjee2021Quantum, Camposeo2025Materials}, a central question has been whether quantum effects can enable superextensive power scaling beyond classical limits. Although the distinction between classical and quantum batteries remains debated \cite{ Andolina2019QuantumVsClassical}, it is widely accepted that charging processes generating entanglement or coherence \cite{Ferraro2018HighPower, Le2018SpinChain, Campaioli2024Colloquium, Andolina2025Genuine, Campaioli2024Colloquium, Zhang2023Dicke, Caravelli2021EnergyStorage, JuliaFarre2020Bounds}, thereby enhancing power, can be regarded as genuinely quantum. Recent studies have explored diverse battery–charger architectures in pursuit of such quantum advantage, which can be broadly classified into two paradigms. In direct-charging protocols \cite{Campaioli2017Enhancing, Rossini2020Quantum, Le2018SpinChain, Grazi2024Controlling, Mondal2022PeriodicallyDriven, Divi2025SYKRandomWalk, Francica2024SYK, Ghosh2020Enhancement, Zhang2019Harmonic}, the battery is initialized in the ground state of its Hamiltonian and then subjected to a sudden quench that drives the system toward excited states. In contrast, in charger-mediated protocols \cite{Ferraro2018HighPower, Farina2019ChargerMediated, Dou2022ExtendedDicke, Andolina2018ExactlySolvable, Crescente2020Ultrafast, Salvia2023RepeatedInteractions, Seah2021Collisional, Dou2022CavityHeisenberg, Yang2024TavisCummings, Crescente2022Mediator}, the battery interacts with an auxiliary quantum system at higher energy, enabling energy transfer until the battery becomes charged.

Only a few models are known to exhibit superextensive power scaling. Notably, atomic–photonic hybrid systems such as the Dicke model and its variants display superextensive behavior when thermodynamic consistency is relaxed and the battery–charger coupling lies in the ultrastrong regime \cite{Ferraro2018HighPower, Crescente2020Ultrafast, Delmonte2021TwoPhoton, Yang2024ThreeLevel}. In this limit, the enhanced coupling mediates energy transfer that scales superextensively with system size. The Dicke model also serves as one of the few experimentally realized platforms demonstrating such enhancement \cite{Quach2022Superabsorption, Maillette2023EnergyTransfers, Hymas2025RoomTemp, Hu2022Superconducting, Qu2023Catalyst, Joshi2022Experimental}. In contrast, within the direct-charging paradigm, the only model known to produce superextensive power is the fermionic Sachdev–Ye–Kitaev (SYK) model. The spin and bosonic counterparts of the SYK model, however, fail to exhibit this scaling \cite{Rossini2020Quantum}. The enhancement in the fermionic case has been attributed to the nonlocal nature of the Jordan–Wigner transformation, which maps the fermionic SYK model to its spin representation \cite{Campaioli2024Colloquium}. It is in this transformed representation that superextensive charging power is obtained.

In this Letter, we show that, as a direct consequence of the $g$-extensivity condition \cite{Arad2016Connecting}, superextensive power scaling cannot occur if $g$ does not scale with system size, even when the charging Hamiltonian is highly nonlocal. Moreover, realizing interactions whose locality scales with the system size is exceedingly difficult in any experimental setting. This naturally leads to a key question: can one achieve superextensive power using a Hamiltonian of constant locality, while keeping both the battery and charger Hamiltonians energy-extensive? In this work, we answer this question affirmatively. We demonstrate that within an experimentally feasible model featuring constant locality and a simple quench protocol it is, in principle, possible to attain superextensive power scaling.

\begin{figure*}[t]
    \centering
    \includegraphics[height=0.4\textwidth]{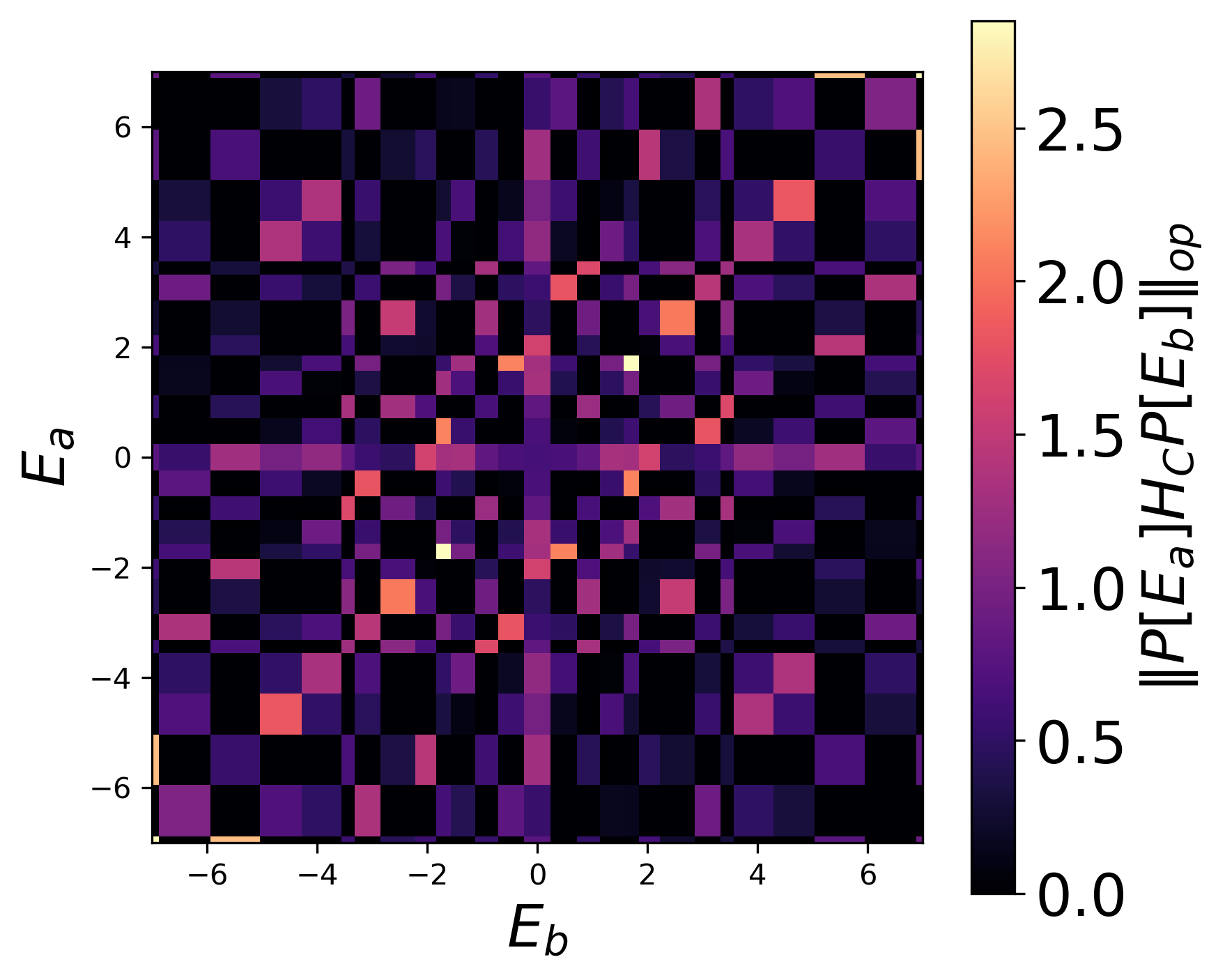}
    \hfill
    \includegraphics[height=0.4\textwidth]{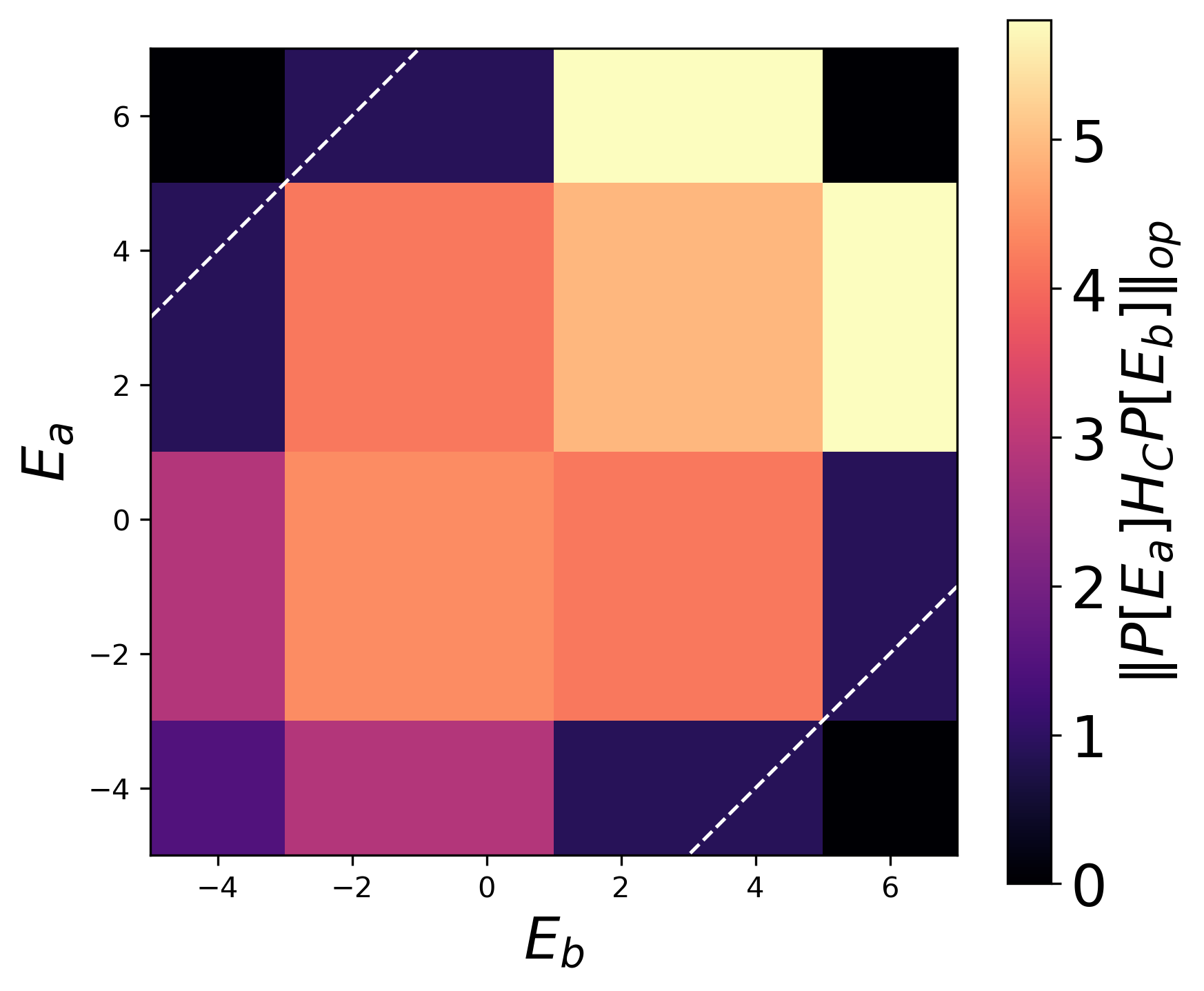}

    \vspace{2pt}
    \makebox[0.48\textwidth][c]{\textbf{(a)}}%
    \hfill
    \makebox[0.48\textwidth][c]{\textbf{(b)}}

    \caption{
        \textbf{Energy Transitions.}
        (a) Heatmap of absolute values of transition matrix elements 
        $\| P[E_a] H_C P[E_b] \|$ for an Ising charger in the diagonal basis 
        of a modified central-spin battery. Degenerate levels are grouped as projections onto energy subspaces of the battery Hilbert space.
        (b) Same as (a), but with the battery and charger Hamiltonians interchanged. Here, total number of sites is seven.
    }
    \label{fig:aklh_heatmap}
\end{figure*}

\prlsection{Interaction Order and $g$-Extensivity of the Hamiltonian}

Here, we consider lattice models \cite{Naaijkens2017QuantumSpin} with local dimension $d$, and total number of sites $N$. The sites are collectively denoted as a set $\Lambda$. The battery and charger Hamiltonians defined on them are $q$-local and $k$-local, respectively. A $k$-local Hamiltonian can be expressed as a sum of local terms \cite{Hastings2010Locality},
\begin{equation}
H = \sum_{X} h_{X},
\end{equation}
where $X \subseteq \Lambda$ denotes subsets of the lattice sites, and $h_{X} = 0$ whenever $|X| > k$. This general structure encompasses both short-range and long-range Hamiltonians with interaction order $k$, since $h_{X}$ need not involve only neighboring lattice sites.  

To quantify the strength of local contributions, we define for any Hamiltonian the parameter $g$ as
\begin{equation}
\max_{i} g_{i} = g, \qquad g_{i} = \sum_{X : X \ni i} \|h_{X}\|.
\end{equation}
This parameter measures the maximum cumulative strength of all local terms acting on a given site and is referred to as the \textit{$g$-extensivity} of the Hamiltonian \cite{Arad2016Connecting}.

As an illustrative example, consider the transverse-field Ising model (TFIM),
\begin{equation}
H = J \sum_{\langle i, j \rangle} \sigma_{i}^{z} \sigma_{j}^{z} + \alpha \sum_{j} \sigma_{j}^{x},
\end{equation}
where $J$ is the interaction strength, $\alpha$ the transverse field, and $\langle i,j \rangle$ denotes nearest-neighbor pairs.  
All the sites $i$ have equal values of $g_{i}$. For site $1$, the $g_{1}$ evaluates to
\begin{align}
g_{1} &= \sum_{X : X \ni 1} \|h_{X}\| 
= \|J \sigma_{1}^{z} \sigma_{2}^{z}\| 
+ \|J \sigma_{1}^{z} \sigma_{N}^{z}\|
+ \|\alpha \sigma_{1}^{x}\| \nonumber\\
&= 2|J| + |\alpha|,
\end{align}
where we have assumed periodic boundary conditions.

\prlsection{Role of Interaction Order} The concept of a quantum battery gained prominence from the idea that charging protocols exploiting quantum resources; particularly entanglement may yield speedups over classical strategies. The underlying physical intuition is that restricting the evolution to the separable subspace severely limits the accessible dynamical pathways, whereas typical states in a generic many-body Hilbert space are highly entangled~\cite{Hayden2006GenericEntanglement, Poulin2011TimeDependent, Hovhannisyan2013WorkExtraction}. Entangling operations can therefore open additional ``routes'' in Hilbert space, enabling faster transitions from low-energy to high-energy states. 

Gyhm \textit{et al.}~\cite{Gyhm2022Quantum} formalized this intuition by studying many-body batteries and identifying constraints on charging power that depend explicitly on the locality, or interaction order, of the charging Hamiltonian. Their results indicate that only \emph{global} operations i.e., interaction order $k$ scaling as $k \sim \mathcal{O}(N)$-- can provide extensive quantum advantage over protocols without entanglement. This aligns with the expectation that highly nonlocal Hamiltonians can generate multipartite entanglement~\cite{Shi2025Multipartite} during charging.

It has further been argued that global operations are necessary to obtain superextensive power. Supporting numerical evidence appeared earlier in Ref.~\cite{Rossini2020Quantum}, where a fermionic complex SYK (cSYK) model was shown to exhibit superextensive charging power. When mapped to spins via the Jordan--Wigner transformation, the resulting Hamiltonian contains highly nonlocal string operators, and it is widely believed \cite{Campaioli2024Colloquium} that these nonlocal terms are responsible for the observed enhancement.

In the next section, we demonstrate that these limitations arise primarily from the assumption that either the battery or the charger Hamiltonian is noninteracting. We show how this restriction can be lifted by allowing both Hamiltonians to be interacting while permitting large values of $g$, yet still preserving energy extensivity for each Hamiltonian.

\begin{figure}[t]
    \centering
    \includegraphics[width=\columnwidth]{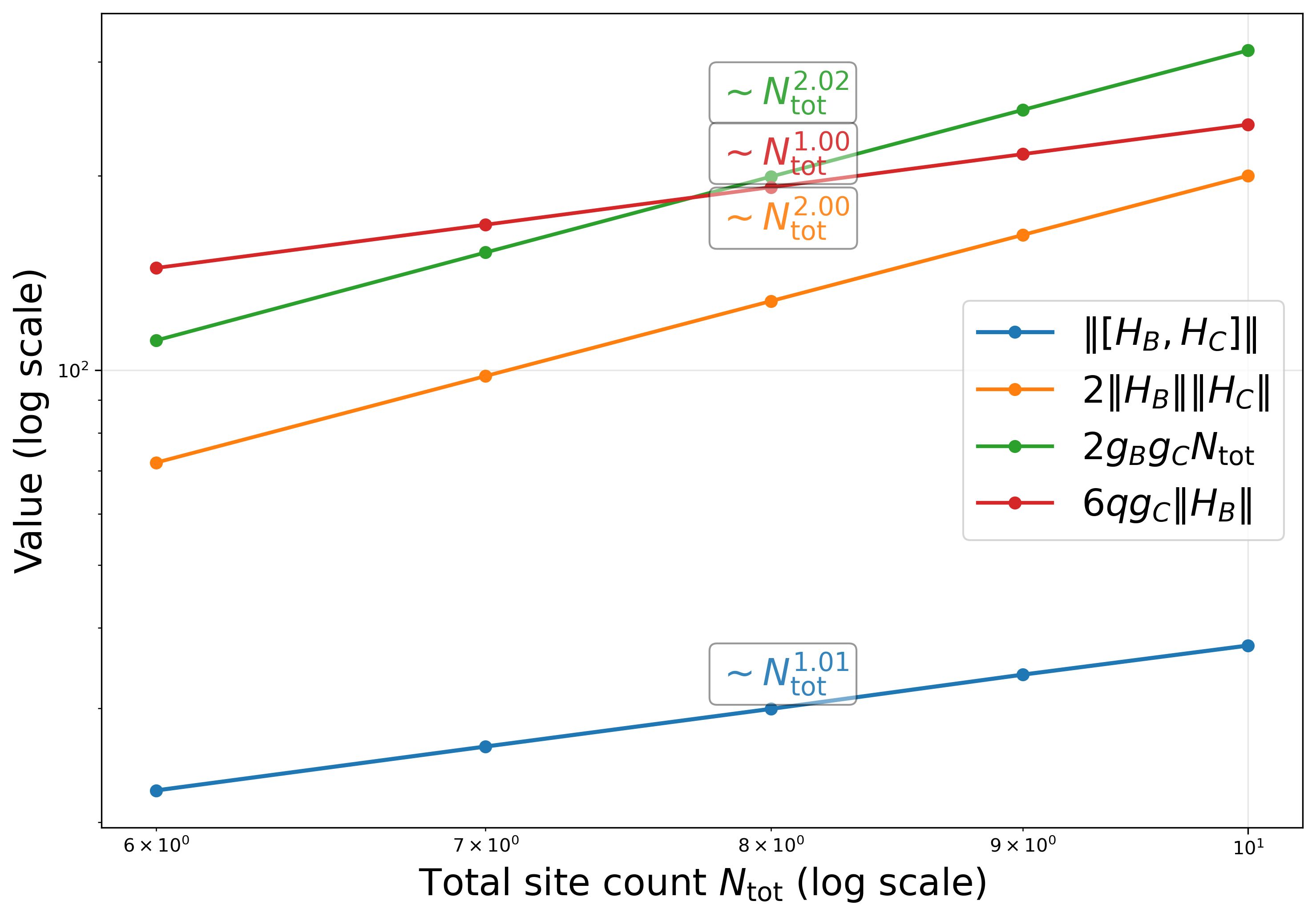}
    \caption{
        Scaling with system size of the commutator operator norm and bounds for an Ising battery and MCS charger
    }
    \label{fig:comm_bounds_Ising_MCS}
\end{figure}

\prlsection{Role of the Battery Hamiltonian}
In the direct charging protocol, the maximum instantaneous power is bounded by the operator norm of the commutator between the battery and charger Hamiltonians,
\begin{equation}
|P(t)| \le \|[H_{B},H_{C}]\|,
\end{equation}
where $\|\cdot\|$ denotes the operator norm. Although the dynamics are generated solely by the charger Hamiltonian $H_{C}$, the bound depends explicitly on $[H_{B},H_{C}]$, showing that the battery Hamiltonian plays an active role in determining the achievable power. This highlights the importance of a joint design of $H_{B}$ and $H_{C}$, in sharp contrast with scenarios where either Hamiltonian is taken to be non-interacting \cite{Campaioli2017Enhancing, Le2018SpinChain, Gyhm2022Quantum}.

Much of the existing literature fixes one of the Hamiltonians, typically the battery or the charger to be non-interacting in order to compare meaningfully with parallel charging scenario. Here we relax this assumption, requiring only that both Hamiltonians are \emph{energy-extensive}, i.e.,
\begin{equation}
\|H_{B}\|,\, \|H_{C}\| \sim \mathcal{O}(N),
\end{equation}
with $N$ the system size. Within this broader framework, we aim to identify explicit, physically motivated models that can exhibit superextensive power scaling.

A key tool in our analysis is the \emph{AKLH lemma}~\cite{Arad2016Connecting}, which bounds matrix elements of local operators between energy subspaces of a $k$-local, $g$-extensive Hamiltonian that are well separated in energy. For an operator $h_{X}$ supported on $X\subseteq\Lambda=\{1,\dots,N\}$, and projectors $\Pi_{<E}$ and $\Pi_{>E'}$ onto the corresponding low- and high-energy subspaces, the lemma asserts
\begin{equation}
\|\Pi_{>E'} h_{X}\Pi_{<E}\|
\le \|h_{X}\|\,
\exp\!\left[-\frac{1}{2gk}\left(E'-E-2R\right)\right],
\end{equation}
where $R\le g|X|$.

As an immediate consequence, if the battery Hamiltonian is noninteracting (i.e., $k_{B}=1$ and $g_{B}=\mathcal{O}(1)$), the commutator is bounded by
\begin{equation}
\|[H_{B},H_{C}]\|\le 6 g_{B} k_{C}\,\|H_{C}\|.
\end{equation}
Exchanging $B$ and $C$ yields the symmetric bound
\begin{equation}
\|[H_{B},H_{C}]\|\le 6 g_{C} k_{B}\,\|H_{B}\|.
\end{equation}
Thus the AKLH lemma does not privilege the charger or the battery; the restriction arises entirely from taking one Hamiltonian to be non-interacting or interacting with all interaction terms pairwise commuting. This symmetry explains why superextensive power is generically forbidden when either $H_{B}$ or $H_{C}$ has bounded locality and bounded $g$~\cite{JuliaFarre2020Bounds, Sarkar2025kLocality} which does not grow with the system size.

To substantiate this viewpoint, we consider a fully interacting battery--charger pair with periodic boundary conditions. The battery is modeled as a central-spin Hamiltonian with Ising interactions among the peripheral spins,
\begin{align*}
H_{B}=\sum_{j=1}^{N}\sigma_{0}^{y}\sigma_{j}^{x}
+ \sum_{j=1}^{N}\sigma_{j}^{z}\sigma_{j+1}^{z},  & \ \ \ \sigma_x^{(N+1)} \equiv \sigma_x^{(1)}.
\end{align*}
while the charger is taken as
\begin{align*}
    H_{C}=\sum_{j=0}^{N}\sigma_{j}^{x}\sigma_{j+1}^{x},  & \ \ \ \sigma_x^{(N+1)} \equiv \sigma_x^{(0)}.
\end{align*}

\begin{figure}[h]
    \centering
    \includegraphics[width=\columnwidth]{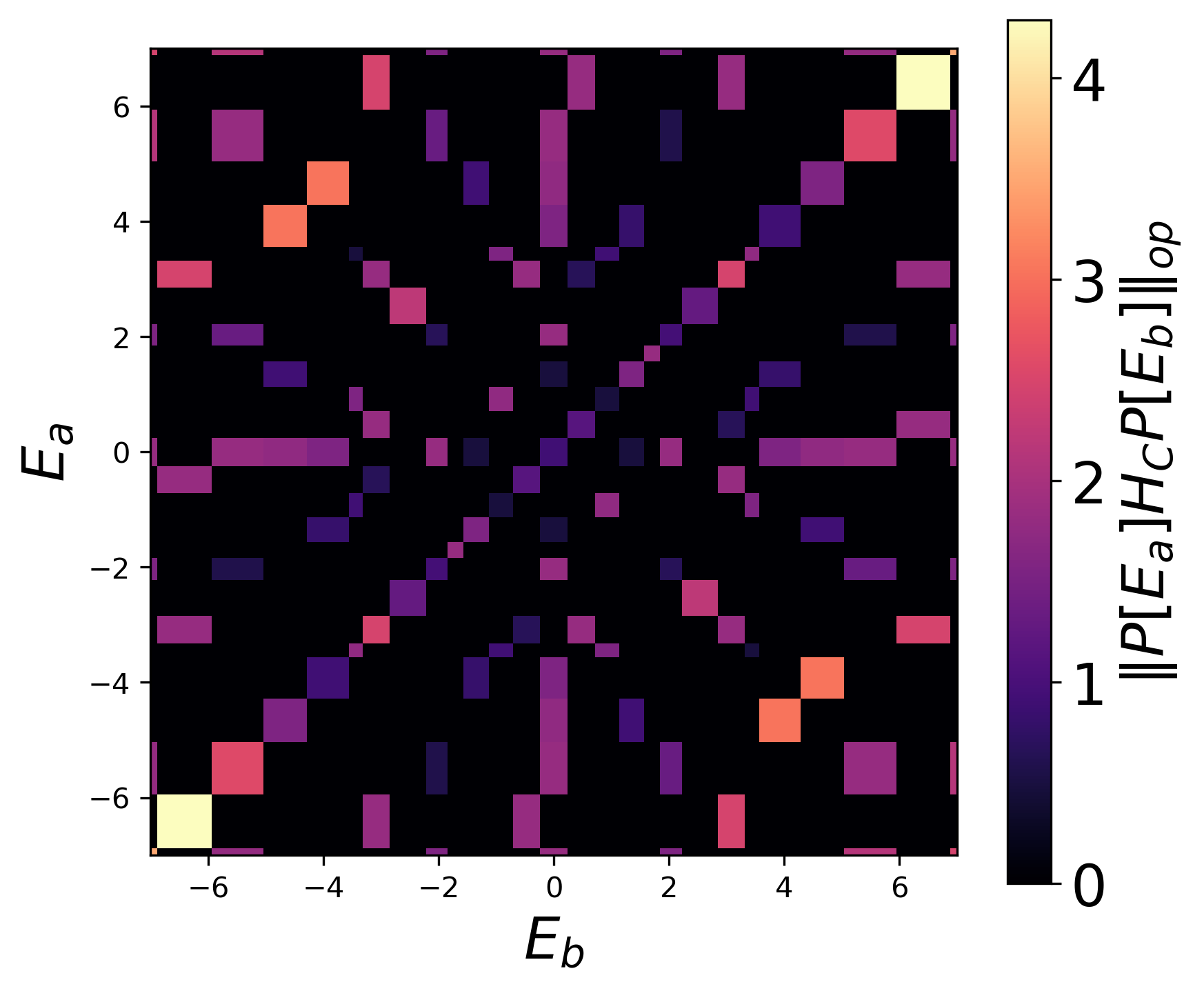}
    \caption{Heatmap of absolute values of transition matrix elements 
        $\| P[E_a] H_C P[E_b] \|$ for the MCS charger in the diagonal basis 
        of the MCS battery. Degenerate levels are grouped as projections onto energy subspaces of the battery Hilbert space.}
    \label{fig:aklh_MCS_MCS}
\end{figure}

For interacting batteries such as this central-spin model, the notion of identical single-site contributions $\|H_{s}\|$ no longer applies uniformly across all sites. Consequently, one cannot use the identical-cell argument \cite{Gyhm2022Quantum} commonly invoked in the non-interacting case. Instead, the extensivity parameter $g$ offers a natural and uniform measure of the overall interaction strength; it coincides with $\|H_{s}\|$ only for batteries composed of identical non-interacting sites.

\begin{figure*}[t]
    \centering
    \includegraphics[height=0.215\textwidth]{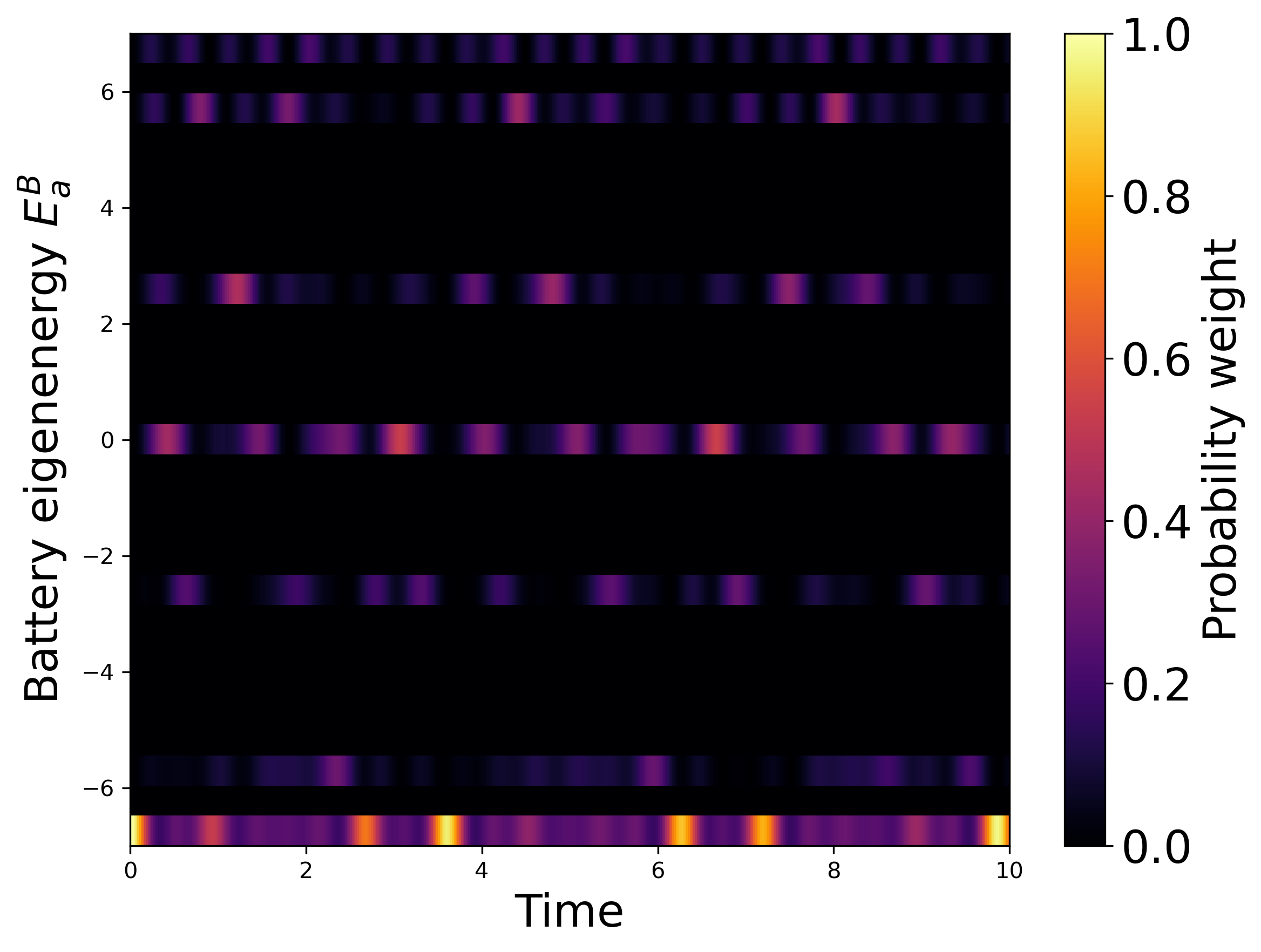}
    \hfill
    \includegraphics[height=0.215\textwidth]{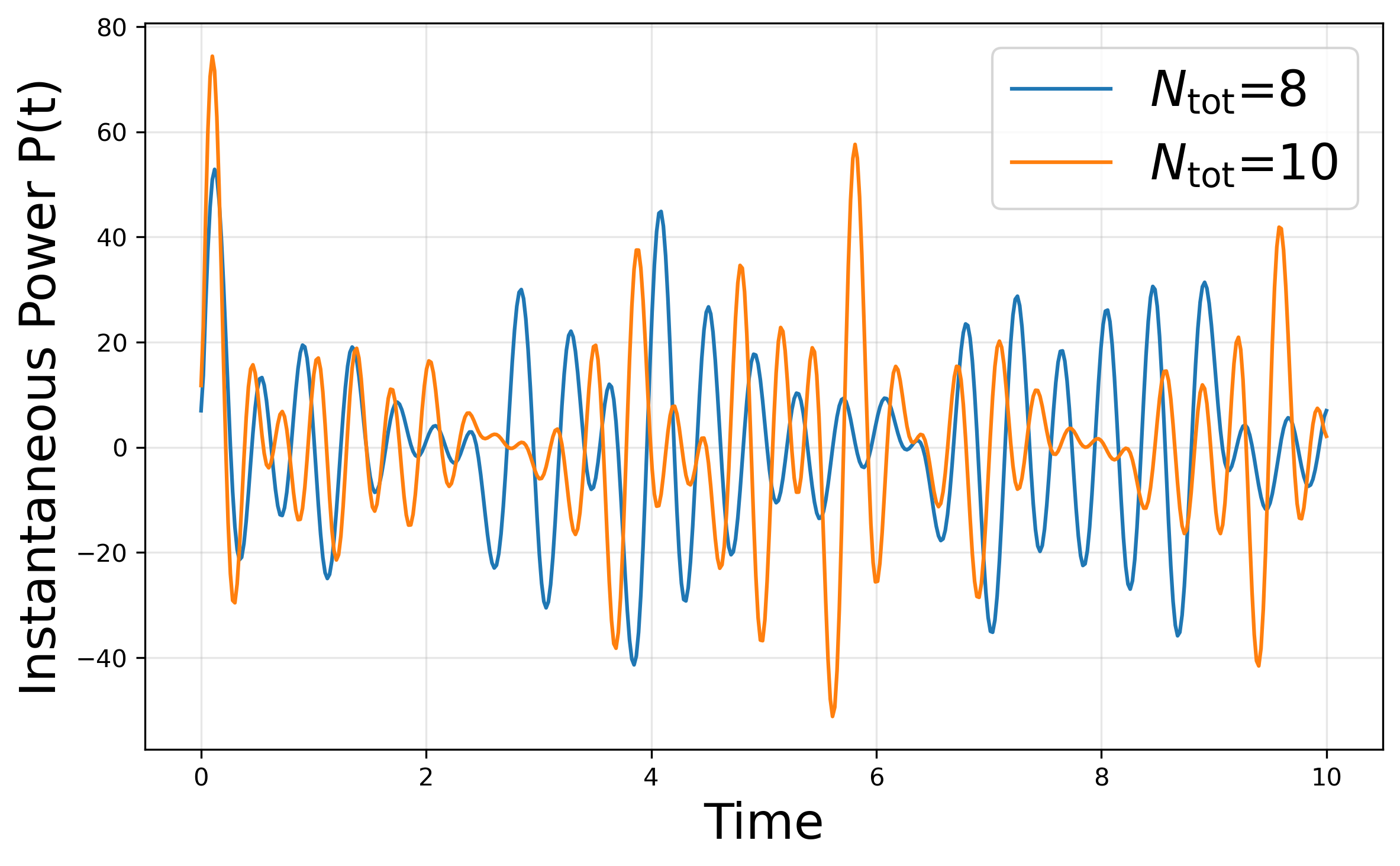}
    \hfill
    \includegraphics[height=0.215\textwidth]{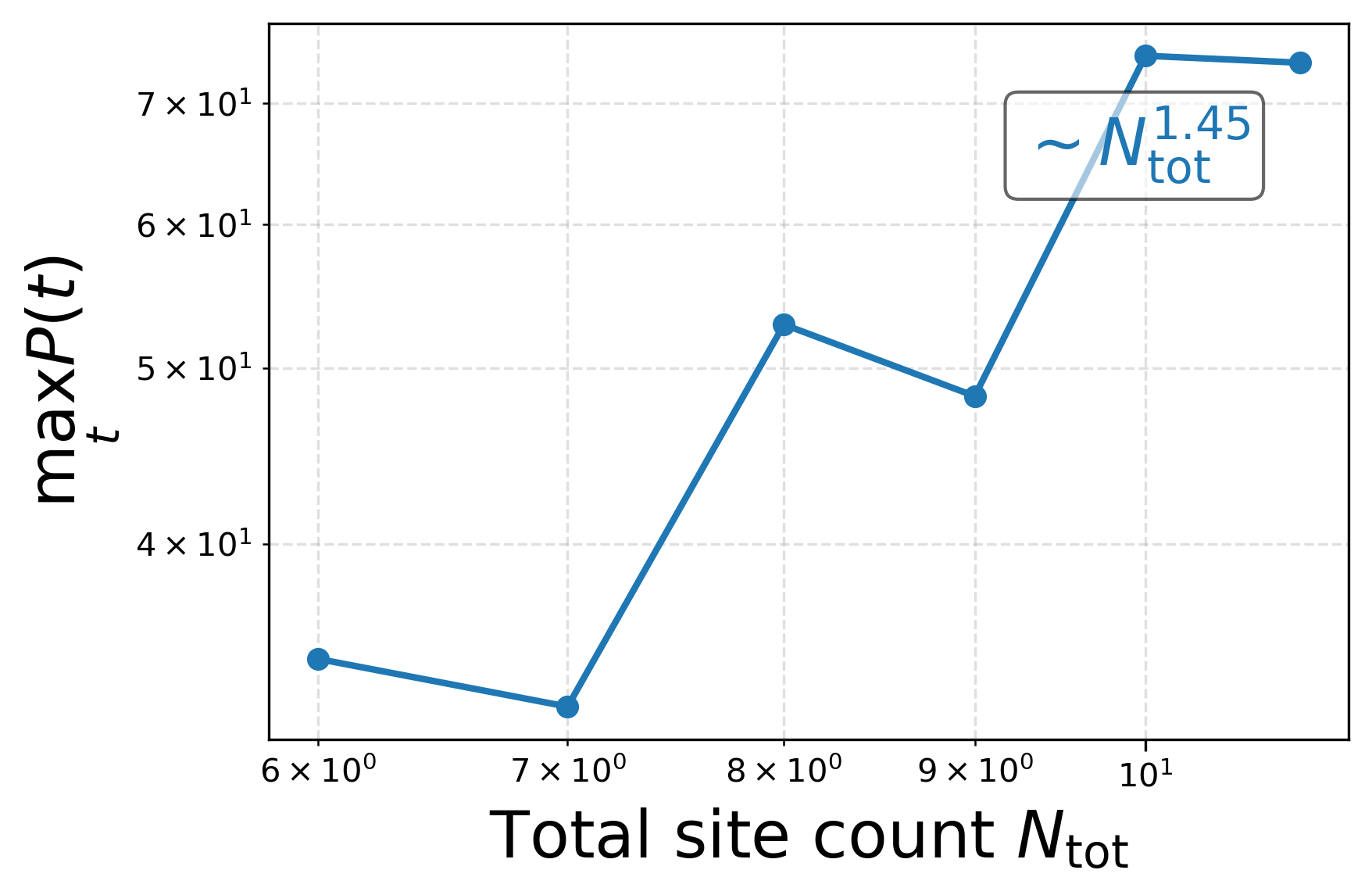}

    \vspace{2pt}
    \makebox[0.3\textwidth][c]{\textbf{(a)}}%
    \hfill
    \makebox[0.3\textwidth][c]{\textbf{(b)}}%
    \hfill
    \makebox[0.3\textwidth][c]{\textbf{(c)}}

    \caption{
        \textbf{Battery Ground State.} For the MCS battery quenched with the MCS charger, we plot: 
        (a) Evolution of the initial ground state with time in the eigenbasis of the Battery,
        (b) The instantaneous power plotted against time,
        (c) The maximum instantaneous power versus total number of sites.
    }
    \label{fig:gs}
\end{figure*}

Although transitions between well-separated energy sectors remain possible when $H_{C}$ is expressed in the eigenbasis of $H_{B}$, exchanging the roles of $H_{B}$ and $H_{C}$ significantly restricts the pattern of energy-sector couplings. Indeed, the AKLH lemma confines energy transitions to a window of width $2g_{C}q$ rather than $2g_{B}k$. As illustrated in Fig.~\ref{fig:aklh_heatmap}a, transitions between well-separated subspaces are generically nonzero, whereas Fig.~\ref{fig:aklh_heatmap}b shows the sharp restriction induced by the $2g_{C}q$ bound (white dashed region). This restriction directly impacts the maximum instantaneous power.

This expectation is borne out in Fig.~\ref{fig:comm_bounds_Ising_MCS}, where we demonstrate numerically that the commutator norm---and hence the maximal instantaneous power---scales as $N_{\mathrm{tot}}^{1.01}$. 

This behavior is consistent with the bound of Eq.~(9), given that $g_{C}=\mathcal{O}(1)$, $k_{B}=2$, and $\|H_{B}\|\sim\mathcal{O}(N)$. Together with our analytical bound~\cite{SupplementaryMaterial}, this numerical result constitutes a central finding of this work: the battery Hamiltonian is not a passive participant but a decisive factor determining the maximal charging power, and interacting battery models can overcome the limitations inherent in non-interacting settings.
 
\prlsection{Superextensive power with low interaction order}
The previous example not only demonstrates the active role of the battery Hamiltonian in determining the maximum attainable power, but also clearly manifests the pivotal role played by $g$. This viewpoint can be further emphasized by the following bound. If both the battery and charger Hamiltonians are $g$-extensive with $g = \mathcal{O}(1)$, i.e.,
\[
H_{B} = \sum_{X} h_{X}^{B}, \qquad 
H_{C} = \sum_{Y} h_{Y}^{C},
\]
then
\begin{equation}
\begin{split}
\|[H_{B},H_{C}]\|
&= \Big\| \Big[\sum_{X} h_{X}^{B}, \sum_{Y} h_{Y}^{C}\Big] \Big\| \\
&\le \sum_{X,Y}\|[h_{X}^{B},h_{Y}^{C}]\|
   \le 2\!\sum_{X,Y}\|h_{X}^{B}\|\,\|h_{Y}^{C}\|.
\end{split}
\end{equation}
Since $[h_{X}^{B},h_{Y}^{C}] = 0$ if $X \cap Y = \emptyset$, we can write
\begin{equation}
\begin{split}
\sum_{X,Y: X\cap Y\neq\emptyset}\!\!\|h_{X}^{B}\|\,\|h_{Y}^{C}\|
&\le \sum_{i=1}^{N}\!\Big(\sum_{X\ni i}\|h_{X}^{B}\|\Big)
       \Big(\sum_{Y\ni i}\|h_{Y}^{C}\|\Big) \\
&\le \sum_{i=1}^{N} g_{B} g_{C}
   = N g_{B} g_{C}.
\end{split}
\end{equation}

If both $g_{B}$ and $g_{C}$ are of order $\mathcal{O}(1)$, then, irrespective of the interaction order, the commutator as well as the maximum instantaneous power can scale only extensively with the system size. This demonstrates that a high participation number alone is insufficient for achieving superextensive power if $g$ does not scale with system size. This behavior is precisely observed in the Lipkin–Meshkov–Glick (LMG) model: despite its all-to-all connectivity, enforcing energy extensivity through proper normalization restricts the maximum power to scale extensively when the battery remains non-interacting. See Ref.~\cite{SupplementaryMaterial} for other models where $g$ dictates the correct scaling of the maximum instantaneous power.

\begin{figure*}[t]
    \centering
    \includegraphics[height=0.215\textwidth]{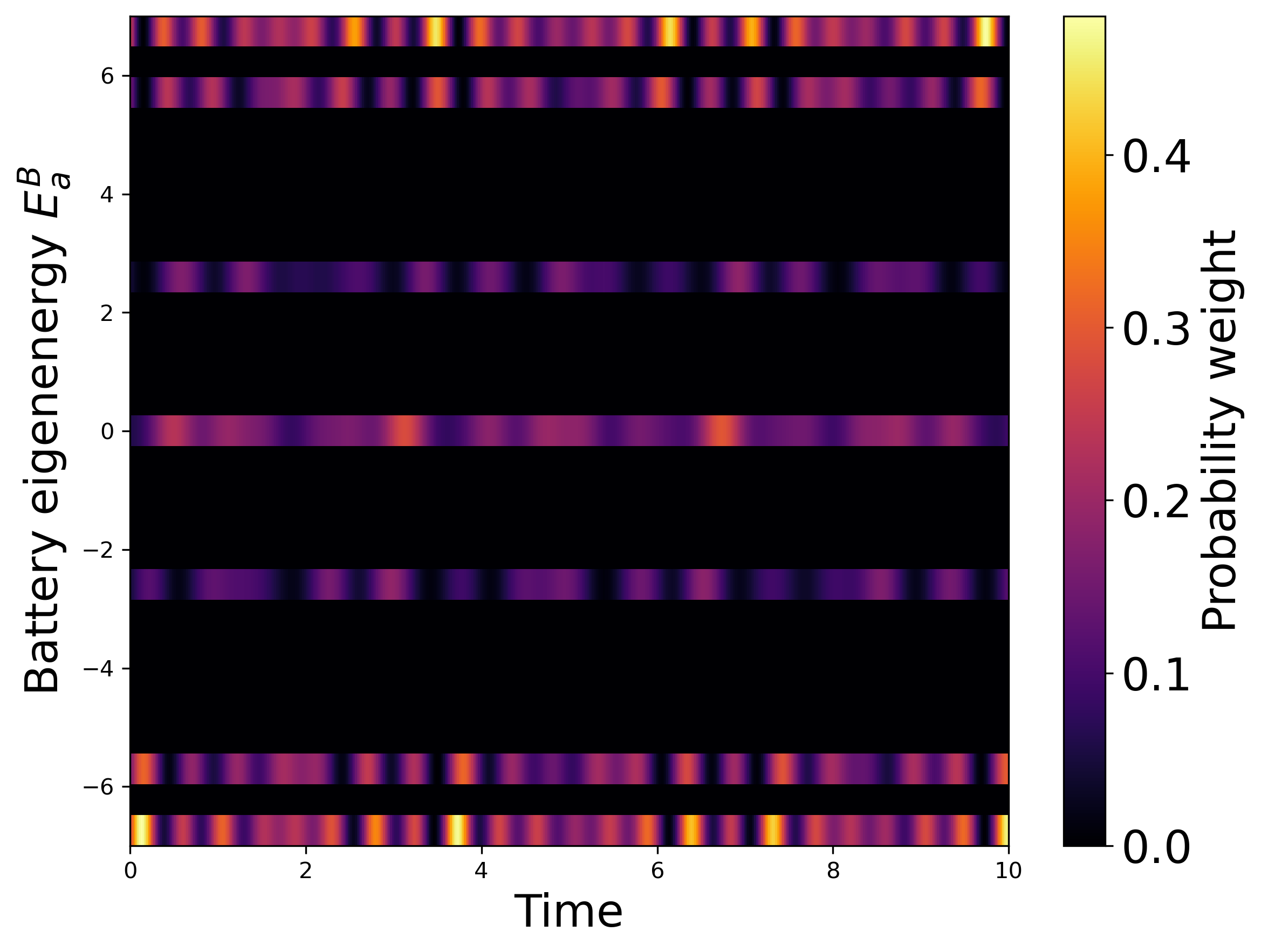}
    \hfill
    \includegraphics[height=0.215\textwidth]{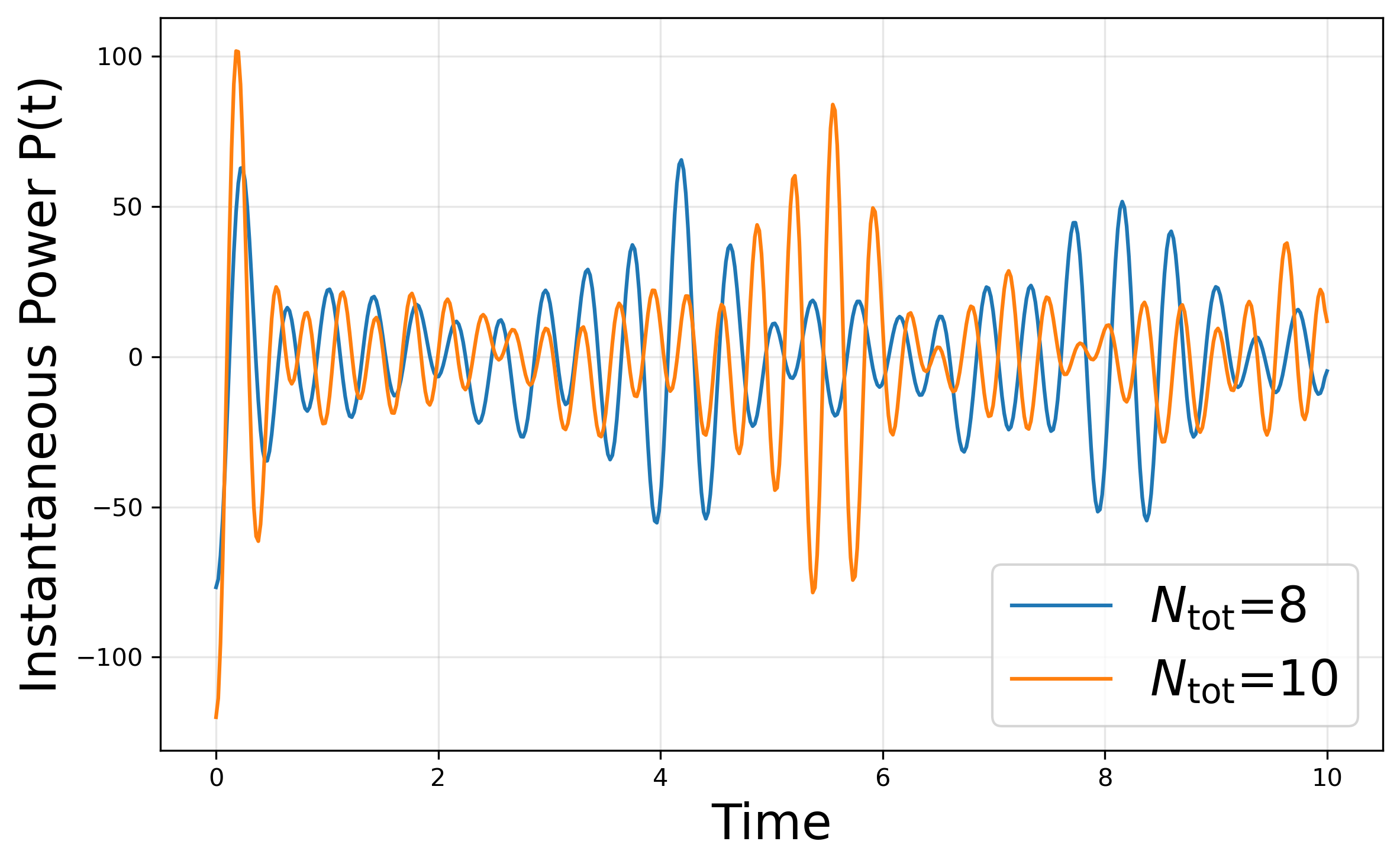}
    \hfill
    \includegraphics[height=0.215\textwidth]{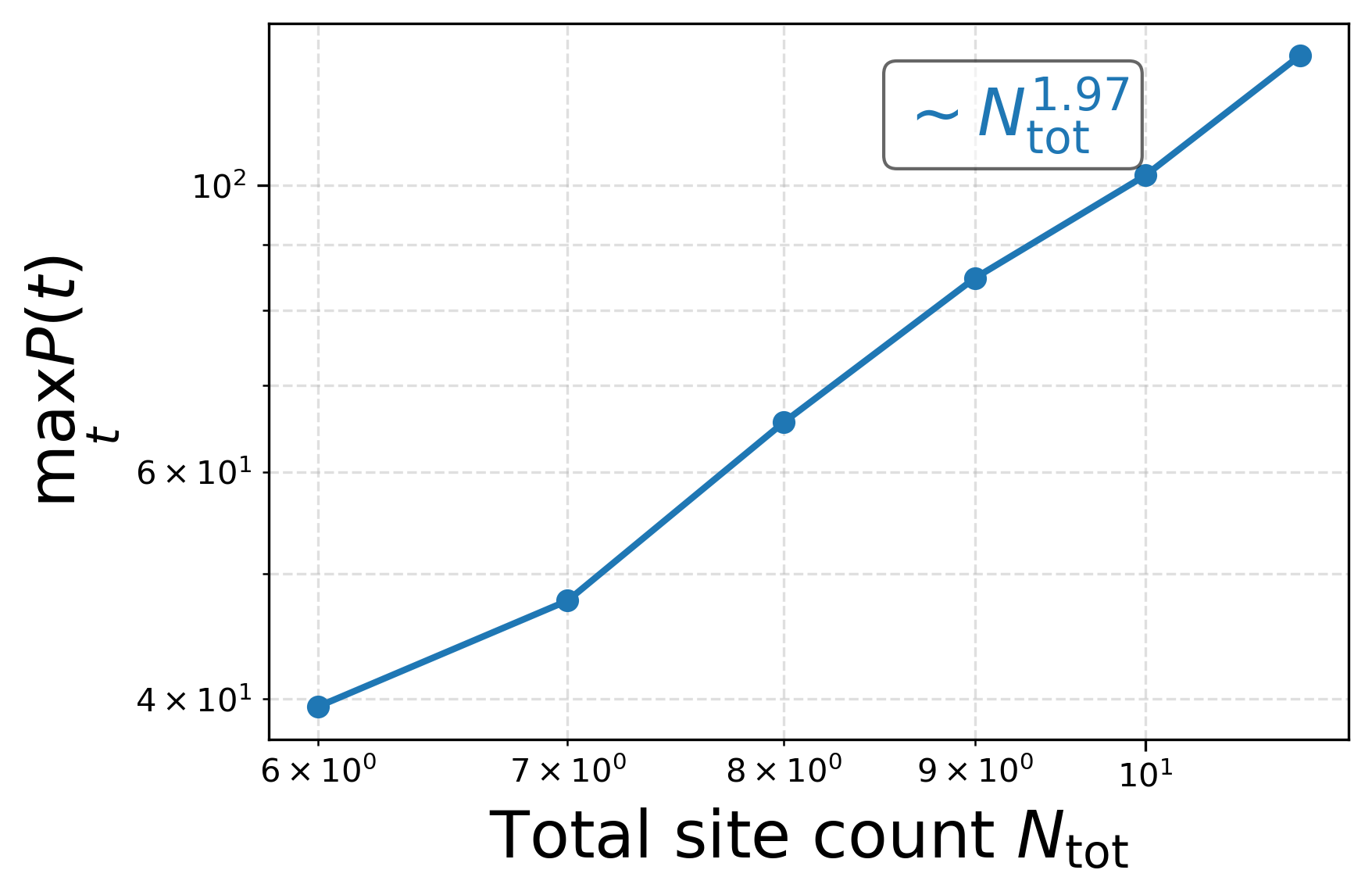}

    \vspace{2pt}
    \makebox[0.3\textwidth][c]{\textbf{(a)}}%
    \hfill
    \makebox[0.3\textwidth][c]{\textbf{(b)}}%
    \hfill
    \makebox[0.3\textwidth][c]{\textbf{(c)}}

    \caption{
        \textbf{Product State.} For the MCS battery quenched with the MCS charger, we plot:
        (a) Evolution of the initial product state with time in the eigenbasis of the Battery,
        (b) The instantaneous power plotted against time,
        (c) The maximum instantaneous power versus total number of sites.
    }
    \label{fig:prod_state}
\end{figure*}

Moreover, this bound hints at the possibility of achieving superextensive power without invoking global operations, while maintaining energy extensivity of both the charger and battery Hamiltonians. The key lies in designing Hamiltonians where the energy budget is distributed unevenly across lattice sites, such that at least one site possesses an interaction energy $g$ that scales with the system size rather than remaining constant.

In Fig. \ref{fig:aklh_heatmap}, the energy transitions are unrestricted when the battery is modeled using the modified central spin model. Even though the locality of the Ising charger is low at $k=2$, the local energy of the central site sets a high $g_B$ (which scales linearly with N) for the battery. Due to this, the AKLH radius within which energy transitions can occur may traverse the entire spectrum.

Choosing both the charger and the battery to be modified central spin Hamiltonians, we can expect that the interchange of Hamiltonians will not tighten bounds on charging power, which theoretically scale as $~ N^2$ in this case. In fact, we numerically confirm a superextensive maximum instantaneous power.

We select the battery Hamiltonian, which is a modified central spin (MCS) model with Ising-type interactions between peripheral spins, as follows:
\begin{align}
H_{\mathrm{B}}
= & \sum_{j=1}^{N}
    \sigma_x^{(0)} \sigma_x^{(j)}
  + \sum_{j=1}^{N}
    \sigma_z^{(j)} \sigma_z^{(j+1)},
\qquad \\
&\sigma_z^{(N+1)} \equiv \sigma_z^{(1)} \; \text{(PBC)}. \nonumber
\end{align}
and quench it with a similar charger Hamiltonian:
\begin{align}
H_{\mathrm{C}}
=& \sum_{j=1}^{N}
    \sigma_y^{(0)} \sigma_x^{(j)}
  + \sum_{j=1}^{N}
    \sigma_z^{(j)} \sigma_z^{(j+1)},
\qquad \\
&\sigma_z^{(N+1)} \equiv \sigma_z^{(1)} \; \text{(PBC)}. \nonumber
\end{align}
The energy transition amplitudes are numerically plotted in Fig. \ref{fig:aklh_MCS_MCS} for this model. As expected, AKLH doesn't restrict transitions to small energy windows, and we have large values of terms far from the diagonal.

Just as in Fig. \ref{fig:aklh_heatmap}, the transitions are not restricted due to small localities. This holds upon an interchange of Hamiltonians, and we therefore may obtain superextensive power.

\prlsection{Initial state dependence of charging dynamics} While the AKLH transition map guarantees the possibility of appreciable transition amplitude between well-separated energy subspaces of the battery Hamiltonian, it does not by itself establish superextensive \emph{charging power}. To verify this directly, we simulate the full dynamics generated by the charger Hamiltonian for several physically relevant initial states.

We begin with the battery ground state, whose time-evolved population dynamics are shown in Fig.~\ref{fig:gs}. The numerical maximum power extracted over the evolution is displayed in Fig.~\ref{fig:gs}(c). Across multiple runs we find superextensive scaling, though the fitted exponent varies between runs; we attribute this variation to the possible ground-state degeneracy of the battery Hamiltonian and the consequent sensitivity to the specific vector chosen within the degenerate subspace. The dataset shown yields an exponent of $\alpha \approx 1.45$, firmly above the extensive threshold.

As a complementary probe, we study a product initial state ($\ket{\psi} = \ket{+_z}_0 \otimes \ket{+_x}^{\otimes N}$), which has support distributed broadly across the battery spectrum. Its dynamics (heatmap of spectral weights, time-dependent power, and scaling of peak power) are shown in Fig.~\ref{fig:prod_state}. In this case the scaling is even clearer: we obtain $\alpha \approx 1.97$, very close to the commutator bound and demonstrating a strong quantum advantage.

Finally, we consider the maximum \emph{attainable} instantaneous power, which is bounded by the commutator spectral norm $\bigl\| [H^{C}, H^{B}] \bigr\|$. Figure~\ref{fig:comm_bounds} shows that this bound itself scales superextensively with system size, approximately as $P_{\max} \propto (N_{\mathrm{tot}})^{1.95}$. This is consistent with the expectation that the optimal state, i.e., the one that attains the operator norm, should exhibit a scaling comparable to the numerical maximum obtained from suitably chosen initial states. For comparison, we also display the two analytic bounds $2\|H_B\|\,\|H_C\|$ and $2 g_B g_C N_{\mathrm{tot}}$, which remain valid but are looser for our model.

\begin{figure}[h]
 \centering
    \includegraphics[width= \columnwidth]{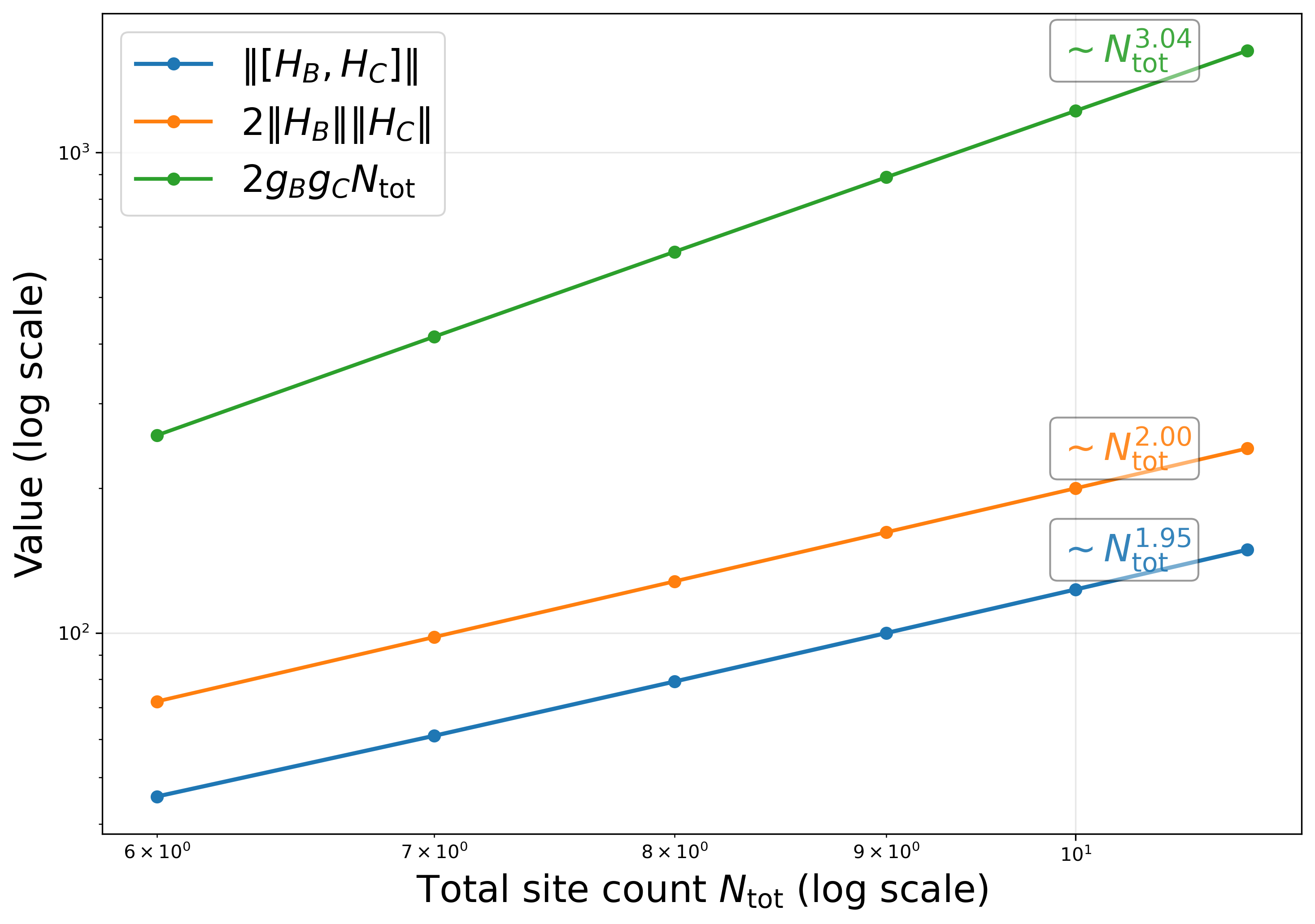}
    \caption{
        Scaling with system size of the commutator operator norm and bounds for an MCS battery and charger
    }
    \label{fig:comm_bounds}
\end{figure}

\prlsection{Discussion}
The key question we address here is whether one can trade the high locality of the battery or charger Hamiltonian for a large value of $g$ when both Hamiltonians are interacting, while still maintaining energy extensivity. Experimental feasibility is also an essential consideration when developing practical quantum batteries. Although many battery--charger architectures may in principle satisfy these criteria, here we focus on a specific and experimentally motivated model that clearly illustrates how $g$-extensivity can enable superextensive power.

Variants of this modified central-spin model have previously appeared in the context of charger-mediated energy transfer \cite{Liu2021CentralSpin, Yang2025CentralSpinSpeedup}. Our setup differs in that we employ the \textit{direct charging protocol}, quenching only the central-spin coupling to the peripheral spins. This model exhibits several properties of independent interest for quantum batteries, such as potentially efficient work extraction and robustness of stored energy~\cite{ Peng2021Bounds, Shi2022Entanglement}. In the present work, however, we focus solely on its ability to generate superextensive power arising from the large value of $g$ for both $H_{B}$ and $H_{C}$. A straightforward calculation~\cite{SupplementaryMaterial} confirms that the commutator of these Hamiltonians grows superextensively with the system size.

An important question that remains is whether quantum effects play a decisive role in these examples. It is well understood that some form of multipartite entanglement is necessary to obtain enhanced power, which is also intuitively clear since a generic $k$-local Hamiltonian can generate $k$-partite entanglement during the dynamics. However, in our case it is \emph{not} the Hamiltonian locality that enables superextensive power; thus the question arises whether the underlying entanglement generation is relevant at all.

To investigate this, we study the dynamics of bipartite entanglement for a system of $7$ spins, where the bipartition consists of the central spin and alternating peripheral spins (in our example $\{0,1,3,5\}$ and $\{2,4,6\}$). The corresponding plots are shown in~\cite{SupplementaryMaterial}. We observe that the growth of entanglement closely mimics the power dynamics. At first sight, this suggests that bipartite entanglement generation may be the mechanism through which higher-energy states are accessed during the dynamics.

For non-interacting batteries, the energy basis can be labelled by local excitations at individual sites, and the charger's interaction order determines the instantaneous energy transition limit $\Delta E$, as well as the source of entanglement generated during evolution. For interacting batteries, the situation is more subtle: the energy eigenstates are entangled, excitations cannot be assigned to individual sites after a sudden quench, and in this case we speculate that the bipartite entanglement growth explicitly depends on~$g$.

Finally, we note that entanglement can suppress locally extractable work \cite{Hokkyo2025ErgotropyBound}. Understanding how a large $g$ influences both the extractable work and the stability of the charged state \cite{Gherardini2020Stabilizing, Bai2020Floquet, PhysRevApplied.14.024092} therefore constitutes an important direction for future research. In addition, extending our framework to other quantum-charging paradigms, such as charger-mediated architectures \cite{Ferraro2018HighPower, Andolina2025Genuine} and Floquet-engineered protocols \cite{Mondal2022PeriodicallyDriven, Puri2024FloquetLongRange}, remains an open and promising avenue for exploration.

\begin{acknowledgments}
We acknowledge D. Rosa, R. K. Shukla and S. Sur for fruitful discussions.
\end{acknowledgments}

\bibliographystyle{apsrev4-2}
\bibliography{bibliography}

\onecolumngrid
\appendix
\section*{Supplemental Material}
\section{Bounding the transition matrix element in commuting case}

\begin{lem}
\cite{Kuwahara2015Thesis} For any $g$-extensive commuting Hamiltonian $H=\sum_{X}h_{X}$ where $[h_{X},h_{X'}]=0~~\forall X,X'$ and \(g\sim\mathcal{O}(1)\). The transition matrix element of another $q$-local operator $h^{(q)}$ acting on sites $S\subset\lambda$ is bounded as follows,
 
\[
\left|\left|\Pi_{\geq E'}h^{(q)}\Pi_{\leq E}\right|\right|=0~~\text{for}~E'-E>2gq.
\]
\end{lem}
\begin{proof}
\begin{equation}
\left|\left|\Pi_{\geq E'}h^{(q})\Pi_{\leq E}\right|\right|=\left|\left|\Pi_{\geq E'}e^{-\lambda H}e^{\lambda H}h^{(q)}e^{-\lambda H}e^{\lambda H}\Pi_{\leq E'}\right|\right|\leq e^{-\lambda(E'-E)}\left|\left|e^{\lambda H}h^{(q)}e^{-\lambda H}\right|\right|,
\end{equation}
Now using the identity,
\[
e^{\lambda H}h^{(q)}e^{-\lambda H}=\sum_{l=0}^{\infty}\frac{-\lambda^{l}}{l!}\mathcal{C}_{l},
\]
where $\mathcal{C}_{0}=h^{(q)}$, \(\mathcal{C}_{1}=[h^{(q)},H]=\sum_{X:X\cap S\neq\emptyset}[h^{(q)},h_{X}]\)...,
    \(\mathcal{C}_{l}=\sum_{X_{1}\cap S\neq\emptyset}\sum_{X_{2}\cap S\neq\emptyset}\cdots\sum_{X_{l}\cap S\neq\emptyset}\left[\cdots\left[\left[h^{(q)},h_{X_{1}}\right],h_{X_{2}}\right]\cdots h_{X_{l}}\right],\)
Using the sub-multiplicative property of the operator norm and the $g$-extensivity of the Hamiltonian we can write,
\begin{equation}
\left|\left|\mathcal{C}_{l}\right|\right|\leq 2^{l}||h^{(q)}||\left(\sum_{X:X\cap S\neq\emptyset}||h_{X}||\right)^{l}\leq 2^{l}||h^{(q)}||(gq)^{l},
\end{equation}
Therefore \[
\|e^{\lambda H}h^{(q)}e^{-\lambda H}\|\le\sum_{l=0}^{\infty}\frac{-\lambda^{l}}{l!}2^{l}\|h^{(q)}\|(gq)^{l}=e^{-2\lambda gq}\|h^{(q)}\|,
\]
Then,
\begin{equation}
    \|\Pi_{\ge E'}h^{(q)}\Pi_{\le E}\|\le e^{-\lambda[(E'-E)-2gq]}\|h^{(q)}\|.
\end{equation}
For, $E'-E>2gq$, \(e^{-\lambda[(E'-E)-2gq]}\|h^{(q)}\|\rightarrow0\) as $\lambda\rightarrow \infty.$
\end{proof}
\section{Commutator bound for commuting Hamiltonians}

In the main text, we considered the central spin model as the battery Hamiltonian 
and the Ising model as the charger Hamiltonian. The corresponding Hamiltonians are
\begin{align*}
H_{B}=\sum_{j=1}^{N}\sigma_{0}^{y}\sigma_{j}^{x}
+ \sum_{j=1}^{N}\sigma_{j}^{z}\sigma_{j+1}^{z},  & \ \ \ \sigma_x^{(N+1)} \equiv \sigma_x^{(1)} \; \text{(PBC)}.
\end{align*}
and
\begin{align*}
    H_{C}=\sum_{j=0}^{N}\sigma_{j}^{x}\sigma_{j+1}^{x},  & \ \ \ \sigma_x^{(N+1)} \equiv \sigma_x^{(0)} \; \text{(PBC)}.
\end{align*}
Since $H_{C}$ is a commuting Hamiltonian, we may directly employ Lemma~1 to 
bound the transition matrix elements entering the commutator 
$\|[H_{B},H_{C}]\|$.  
To generalize this argument for any battery and charger Hamiltonians, 
we consider an arbitrary commuting Hamiltonian $H^{(c)}$ and another 
$k$-local Hamiltonian $H$.  
This bound has appeared in the literature \cite{Kuwahara2016Exponential, Sarkar2025kLocality}; however, for completeness, we provide
the full derivation below.

\subsection{Discretizing the spectrum of the commuting Hamiltonian}

We discretize the spectrum of the commuting Hamiltonian $H^{(c)}$ as
\begin{equation}
    H^{(c)'}=\sum_{m=-\infty}^{\infty}\left(m\epsilon+\frac{\varepsilon}{2}\right)
    \Pi_{[m\varepsilon,m\varepsilon+\varepsilon)},
\end{equation}
where $\Pi_{[m\varepsilon,m\varepsilon+\varepsilon)}$ is the projector onto the
eigenspace of $H^{(c)}$ whose eigenvalue lies within the interval 
$[m\varepsilon,m\varepsilon+\varepsilon)$.  
We denote 
\[
    \delta H^{(c)}= H^{(c)} - H^{(c)'}.
\]
By construction, 
\[
    \|\delta H^{(c)}\| \leq \frac{\varepsilon}{2},
\]
which is the maximum discretization error arising from replacing 
the actual spectrum with the nearest discretized values.

We now write
\begin{equation}
\label{eq:comm_split}
    \|[H^{(c)},H]\|
    \le \|[H^{(c)'},H]\| + \|[\delta H^{(c)},H]\|.
\end{equation}
By sub-multiplicative property of the operator norm,
\[
    \|[\delta H^{(c)},H]\|\le 2\|\delta H^{(c)}\|\,\|H\|\le \varepsilon\|H\|.
\]

\subsection{Bounding the commutator using transition elements}

Using the projectors $\Pi_{m}\equiv \Pi_{[m\varepsilon,m\varepsilon+\varepsilon)}$,
we write
\begin{align}
    \|[H^{(c)'},H]\|=\max_{\ket{\psi}\in\mathcal{H}}\left|\bra{\psi}[H^{(c)'},H]\ket{\psi}\right|
    &= \max_{\ket{\psi}}\left|
    \sum_{m,m'}\bra{\psi}\Pi_{m}(H^{(c')}H-HH^{(c')})\Pi_{m'}\ket{\psi}\right|
    \\
    &=\max_{\ket{\psi}}\left|\sum_{m,m'}\epsilon(m'-m)\bra{\psi}\Pi_{m}H\Pi_{m'}\ket{\psi}\right|
     \\
    &\le
      \sum_{m,m'}\epsilon|m'-m|\|\bra{\psi}\Pi_{m}\| \,
    \|\Pi_{m} H \Pi_{m'}\|\,\|\Pi_{m'}\ket{\psi}\|.
\end{align}
Denoting, \(\|\bra{\psi}|\Pi_{m}\|=\alpha_{m}\) and \(\|\Pi_{m'}\ket{\psi}\|=\alpha_{m'}\), Lemma~1 implies that 
\[
    \|\Pi_{m}H\Pi_{m'}\|= 0
    \qquad \text{if}\qquad
    m\epsilon-(m'\epsilon+\epsilon) >{2gk}.
\]
Hence only the terms with 
\[
    |m-m'|\le \frac{2gk}{\varepsilon}+1
\]
contribute.

Thus,
\begin{align}
    \|[H^{(c)'},H]\|
    &\le \varepsilon\|H\|
    \sum_{\substack{m,m'\\ |m-m'|\le \frac{2gk}{\epsilon}+1}}
    |m'-m| \alpha_{m}\alpha_{m'} \\
    &\le \varepsilon\|H\|
    \sum_{|m-m'|\le \frac{2gk}{\varepsilon}+1}
    \frac{|m'-m|}{2}\left(\alpha_{m}^{2}+\alpha_{m'}^{2}\right) \\
    &\le\epsilon\|H\|\sum_{|m-m'|\le\frac{2gk}{\epsilon}+1}|m-m'|\alpha_{m'}^{2}\\
    &\le \varepsilon\|H\|
    \sum_{m=-\lfloor{\frac{2gk}{\varepsilon}+1}\rfloor}^{\lfloor{\frac{2gk}{\varepsilon}+1}\rfloor}|m|\sum_{m'}\alpha_{m'}^{2}.
\end{align}

Since $\sum_{m=-L}^{L}|m| = L(L+1)$ and $\sum_{m'}\alpha_{m'}^{2}=1$, we obtain
\[
    \|[H^{(c)'},H]\|
    \le \varepsilon\|H\|
    \left(\frac{2gk}{\varepsilon}\right)
    \left(\frac{2gk}{\varepsilon}+1\right).
\]

Including the earlier error term from \eqref{eq:comm_split}, we get
\[
    \|[H^{(c)},H]\|
    \le \epsilon||H||\left(\frac{2gk}{\varepsilon}\right)
    \left(\frac{2gk}{\varepsilon}+1\right)+\epsilon||H||.
\]
The optimal value of $\epsilon$ is $2gk+\delta\epsilon,~\delta\epsilon>0$, which minimizes the right hand side of the inequality.
Taking $\delta\epsilon\to 0$, the optimal bound becomes
\begin{equation}
    \|[H^{(c)},H]\| \le 6gk\,\|H\|.
\end{equation}

Thus, the commutator between a commuting $g$-extensive Hamiltonian 
and any $k$-local Hamiltonian is bounded by a quantity proportional to $gk$.

\section{Examples of the $6 g k \|H\|$ Bound}

We now discuss the results reported in \cite{Le2018SpinChain} in light of Proposition~1 of our main text, illustrating how it correctly predicts the upper bound scaling of the maximum power observed numerically in that work. The battery Hamiltonian considered therein is given by
\begin{equation}
    H_{B} = B \sum_{i=1}^{N} \sigma_{i}^{z}
    - \sum_{i<j} g_{ij} \left[ \sigma_{i}^{z}\sigma_{j}^{z}
    + \alpha\left( \sigma_{i}^{x}\sigma_{j}^{x} + \sigma_{i}^{y}\sigma_{j}^{y} \right) \right],
\end{equation}
where the coupling strength is defined as
\begin{equation}
    g_{ij} =
    \begin{cases}
        g\,\delta_{i,j-1}, & \text{for nearest-neighbour interactions}, \\
        \dfrac{g}{|i-j|^{p}}, & \text{for long-range interactions}.
    \end{cases}
\end{equation}

\vspace{0.5em}
\noindent\textbf{Nearest-neighbour case.}
For $g_{ij} = g\delta_{i,j-1}$, the maximum onsite energy is
\[
g_{B} = B + 2g(1 + 2\alpha),
\]
which scales as $\mathcal{O}(1)$, while the operator norm of the Hamiltonian scales as $\|H_{B}\| = \mathcal{O}(N)$. The charger Hamiltonian $H_C$ is taken to be non-interacting. In this situation, one finds $k = 1$ and $q = 2$, yielding a maximum average power of order $\mathcal{O}(N)$. Since $g_{c} \sim \mathcal{O}(1)$ when $\|H_{C}\| \sim \mathcal{O}(N)$, we obtain
\[
\, g_{c}\|H_{B}\| \sim \mathcal{O}(N),
\]
which indicates extensive scaling of the maximum achievable power.

\vspace{0.5em}
\noindent\textbf{Long-range case with $p=1$.}
For Coulomb-like interactions, $g_{ij} = g/|i-j|$, the onsite energy becomes
\[
g_{B} = B + g\left( 1 + \frac{1}{2} + \frac{1}{3} + \cdots + \frac{1}{N} \right)
+ g\alpha\left( 1 + \frac{1}{2} + \frac{1}{3} + \cdots + \frac{1}{N} \right),
\]
which scales as $\mathcal{O}(\ln N)$. In this case, $\|H_{B}\| \sim \mathcal{O}(N\ln N)$, and we have:
\[
 \, g_{c}\|H_{B}\| \sim \mathcal{O}(N\ln N).
\]
This logarithmic enhancement of power scaling with system size agrees precisely with the numerical results reported in \cite{Le2018SpinChain}.

\vspace{0.5em}
\noindent\textbf{Infinite-range case ($p=0$).}
When $p=0$, the model reduces to an all-to-all coupled system, where each spin interacts with every other spin. Consequently, each site contributes to $N$ pairwise terms, giving $g_{B} \sim \mathcal{O}(N)$ while $\|H_{B}\| \sim \mathcal{O}(N^{2})$. In this limit, the upper bound $6 g_C q \|H_B\|$ of the maximum extractable power scales quadratically with system size, $\mathcal{O}(N^{2})$, consistent with the scaling behaviour reported in \cite{Le2018SpinChain}. This example clearly highlights the significance of the parameter $g$, or equivalently, the onsite energy associated with the battery or charger Hamiltonian, in determining the power scaling behaviour rather than the participation number that was previously emphasized in the literature~\cite{Campaioli2017Enhancing}. Following the same reasoning, it becomes evident why the Lipkin–Meshkov–Glick (LMG) Hamiltonian, despite its all-to-all coupling structure, cannot yield super-extensive power when the Hamiltonian is required to remain energy-extensive. In the LMG model, although the participation number $m$ increases with system size, the onsite energy remains bounded as $\mathcal{O}(1)$. Consequently, the maximum achievable power scales only extensively with $N$, in accordance with the constraints of energy extensivity.

\section{Commutator Evaluation between $H_B$ and $H_C$}
For the following battery and charger Hamiltonians, respectively:
\begin{align}
H_B &= \sum_{i=1}^{N} \sigma_x^{(0)} \sigma_x^{(i)} 
   + \sum_{j=1}^{N} \sigma_z^{(j)} \sigma_z^{(j+1)}, \\[4pt]
H_C &= \sum_{j=1}^{N} \sigma_y^{(0)} \sigma_x^{(j)} 
   + \sum_{j=1}^{N} \sigma_z^{(j)} \sigma_z^{(j+1)}.
\end{align}
Note that $\|H_B\|,\|H_C\| \sim O(N)$, even though the values are not exactly $N$ (which is numerically set to $N$ for the plots by including a $O(1)$ normalization factor in the main text). We write $H_B = H_{CS} + H_{Ising}$ and $H_C = H_{CS}^{'} + H_{Ising}$, with
\begin{align}
H_{CS} &= \sum_i \sigma_x^{(0)} \sigma_x^{(i)}, \qquad
H_{CS}^{'} = \sum_j \sigma_y^{(0)} \sigma_x^{(j)}, \qquad
H_{Ising} = \sum_k \sigma_z^{(k)} \sigma_z^{(k+1)}.
\end{align}
Then,
\begin{equation}
[H_B, H_C] = [H_{CS},H_{CS}^{'}] + [H_{CS},H_{Ising}] + [H_{CS}^{'},H_{Ising}],
\end{equation}
since $[H_{Ising},H_{Ising}]=0$.
For the first term,
\begin{align}
[H_{CS},H_{CS}^{'}]
&= \sum_{i,j} [\sigma_x^{(0)}\sigma_x^{(i)},\, \sigma_y^{(0)}\sigma_x^{(j)}] 
= \sum_{i,j} [\sigma_x^{(0)},\sigma_y^{(0)}]\,\sigma_x^{(i)}\sigma_x^{(j)} \notag\\
&= 2i\,\sigma_z^{(0)} \sum_{i,j} \sigma_x^{(i)}\sigma_x^{(j)}
= 2i\,\sigma_z^{(0)}\!\left(\sum_{i=1}^{N}\sigma_x^{(i)}\right)^{\!2}.
\end{align}
This dominates the scaling of the commutator's operator norm at $O(N^2)$.
For the nearest-neighbour Ising part,
\begin{align}
[H_{CS},H_{Ising}] 
&= \sum_{i,k} [\sigma_x^{(0)}\sigma_x^{(i)},\, \sigma_z^{(k)}\sigma_z^{(k+1)}] \notag\\
&= -2i\,\sigma_x^{(0)} 
   \sum_{i=1}^{N}\!\left(
      \sigma_y^{(i)}\sigma_z^{(i+1)} 
      + \sigma_z^{(i-1)}\sigma_y^{(i)}
   \right),
\\[4pt]
[H_{CS}^{'},H_{Ising}] 
&= \sum_{k,j} [\sigma_y^{(0)}\sigma_x^{(j)} ,\, \sigma_z^{(k)}\sigma_z^{(k+1)}] \notag\\
&= - 2i\,\sigma_y^{(0)} 
   \sum_{k=1}^{N}\!\left(
      \sigma_y^{(k)}\sigma_z^{(k+1)} 
      + \sigma_z^{(k)}\sigma_y^{(k+1)}
   \right).
\end{align}
Both $[H_{CS},H_{Ising}]$ and $[H_{CS}^{'},H_{Ising}]$ have contributions to the operator norm that scale linearly with $N$.
Collecting all parts,
\begin{align}
[H_B,H_C]
&= 2i\,\sigma_z^{(0)}\!\left(\sum_{i=1}^{N}\sigma_x^{(i)}\right)^{\!2} \notag\\
&\quad - 2i\,\sigma_x^{(0)} 
   \sum_{i=1}^{N}\!\left(
      \sigma_y^{(i)}\sigma_z^{(i+1)} 
      + \sigma_z^{(i-1)}\sigma_y^{(i)}
   \right) \notag\\
&\quad - 2i\,\sigma_y^{(0)} 
   \sum_{i=1}^{N}\!\left(
      \sigma_y^{(i)}\sigma_z^{(i+1)} 
      + \sigma_z^{(i)}\sigma_y^{(i+1)}
   \right).
\end{align}
To see that the first term, proportional to $\sigma_z^{(0)}(\sum_i \sigma_x^{(i)})^2$, dominates the operator norm as $\|[H_B,H_C]\|\sim \mathcal{O}(N^2)$, while the remaining terms scale as $\mathcal{O}(N)$, consider the product state $\ket{\psi} = \ket{+_z}_0 \otimes \ket{+_x}^{\otimes N}$, and the inequality:
\begin{align}
    \|[H_B,H_C]\| \geq |\bra{\psi}[H_B,H_C]\ket{\psi}| = |2 i N^2| = 2 N^2.
\end{align}
Thus, the commutator exhibits super-extensive growth due to the collective central–spin interaction.

\section{Variance of $H_B$ and $H_C$ for the product state}
We consider the product state
\[
\ket{\psi} = \ket{+_z}_0 \otimes \ket{+_x}^{\otimes N},
\]
and the Hamiltonians
\[
H_B = \sum_{i=1}^{N} \sigma_x^{(0)}\sigma_x^{(i)} 
     + \sum_{j=1}^{N}\sigma_z^{(j)}\sigma_z^{(j+1)}, \qquad
H_C = \sum_{k=1}^{N}\sigma_y^{(0)}\sigma_x^{(k)} 
     + \sum_{j=1}^{N}\sigma_z^{(j)}\sigma_z^{(j+1)}.
\]
\paragraph*{Expectation values.}
Since $\braket{+_z|\sigma_x|+_z} = \braket{+_z|\sigma_y|+_z} = 0$ and 
$\braket{+_x|\sigma_z|+_x}=0$, all linear terms vanish, giving
\[
\braket{H_B}_\psi = \braket{H_C}_\psi = 0.
\]
\paragraph*{Second moments.}
Decompose $H_B =  H_{CS} + H_{Ising}$ with 
$H_{CS} = \sum_i \sigma_x^{(0)}\sigma_x^{(i)}$ and 
$ H_{Ising} = \sum_j \sigma_z^{(j)}\sigma_z^{(j+1)}$. 
Using $(\sigma_x^{(0)})^2 = \mathbb{I}$,
\[
H_{CS}^2 = \sum_{i,k = 1}^{N}\sigma_x^{(i)}\sigma_x^{(k)} 
\quad\Rightarrow\quad 
\braket{H_{CS}^2}_\psi = \sum_{i,k= 1}^{N} 
\braket{+_x^{\otimes N}|\sigma_x^{(i)}\sigma_x^{(k)}|+_x^{\otimes N}} = N^2,
\]
since $\braket{+_x|\sigma_x|+_x}=1$.
For the ring Ising term,
\[
H_{Ising}^2 = \sum_{j,\ell= 1}^{N}\sigma_z^{(j)}\sigma_z^{(j+1)}\sigma_z^{(\ell)}\sigma_z^{(\ell+1)},
\]
and only identical bonds contribute on $\ket{+_x}^{\otimes N}$, yielding
$\braket{H_{Ising}^2}_\psi = N$.
Mixed terms vanish, $\braket{H_{CS}H_{Ising}}_\psi = 0$, because each contains an unpaired
$\sigma_z$ on some site with zero mean. Hence
\[
\braket{H_B^2}_\psi = N^2 + N, 
\qquad \mathrm{Var}_\psi(H_B) = N^2 + N.
\]
\paragraph*{Variance of $H_C$.}
Similarly, write $H_C = H_{CS}^{'} + H_{Ising}$ with 
$ H_{CS}^{'} = \sum_k \sigma_y^{(0)}\sigma_x^{(k)}$.
Since $(\sigma_y^{(0)})^2=\mathbb{I}$ and 
$\braket{+_x|\sigma_x|+_x}=1$,
\[
\braket{( H_{CS}^{'})^2}_\psi = N^2, \qquad 
\braket{H_{Ising}^2}_\psi = N, \qquad 
\braket{ H_{CS}^{'}H_{Ising}}_\psi = 0.
\]
Therefore
\[
\braket{H_C^2}_\psi = N^2 + N, 
\qquad \mathrm{Var}_\psi(H_C) = N^2 + N.
\]
Therefore, for the chosen product state,
\[
\mathrm{Var}_\psi(H_B) = \mathrm{Var}_\psi(H_C) = N^2 + N,
\]
implying both variances scale quadratically with system size:
\[
\mathrm{Var}_\psi(H_B),\ \mathrm{Var}_\psi(H_C) \sim \mathcal{O}(N^2).
\]
This scaling matches that of the commutator norm 
$\|[H_B,H_C]\|\sim N^2$, showing that this state already exhibits
super-extensive fluctuations.

\section{Power and Rate of Entanglement Growth}
We numerically study the dynamics of bipartite entanglement growth in a system with $N_{tot} = 7$, where the bipartition includes the center and alternating peripheral sites. The plots obtained are shown in Figures \ref{fig:entropy}.

\begin{figure*}[t]
    \centering
    \includegraphics[width=0.48\textwidth]{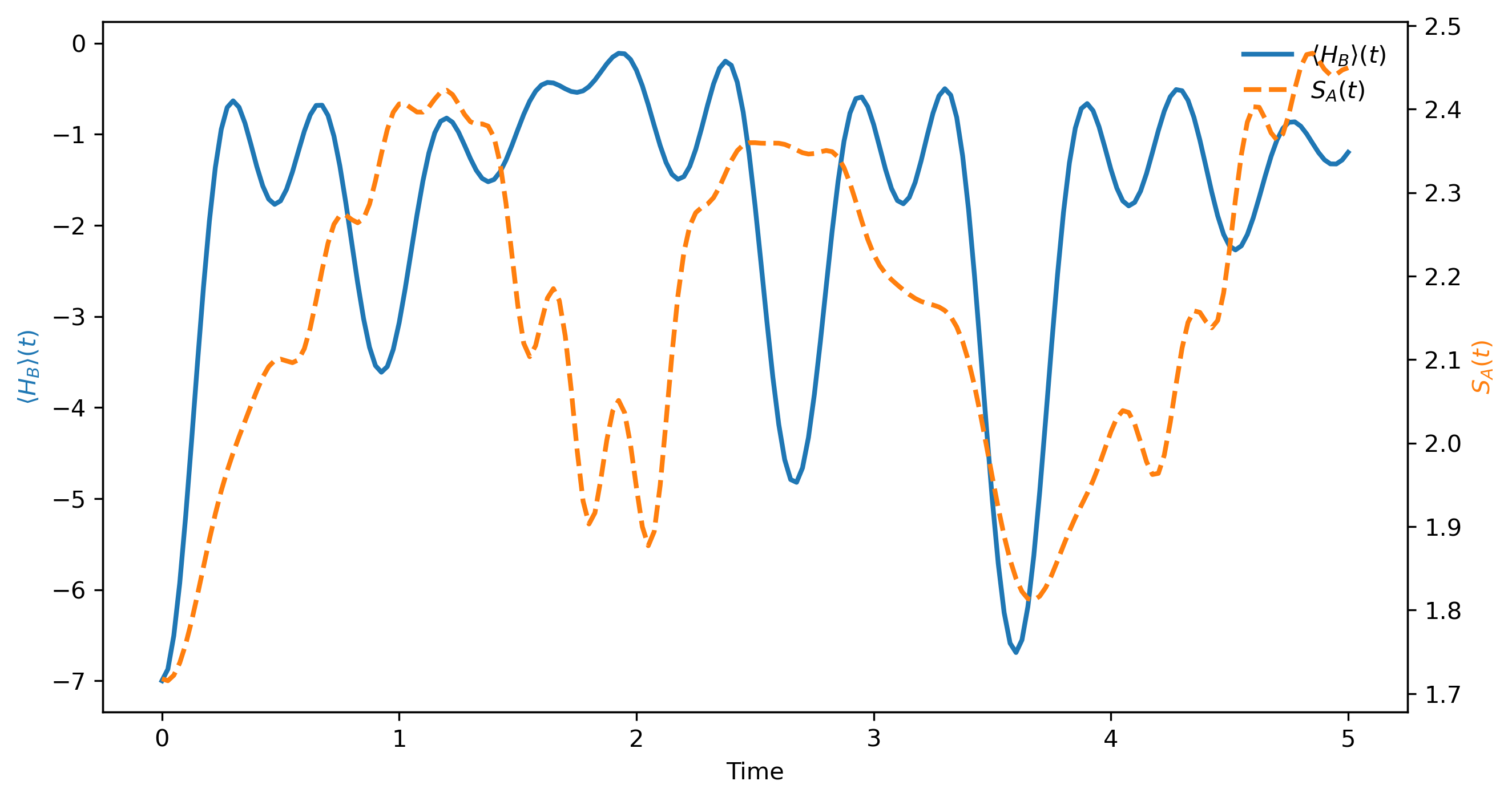}
    \hfill
    \includegraphics[width=0.48\textwidth]{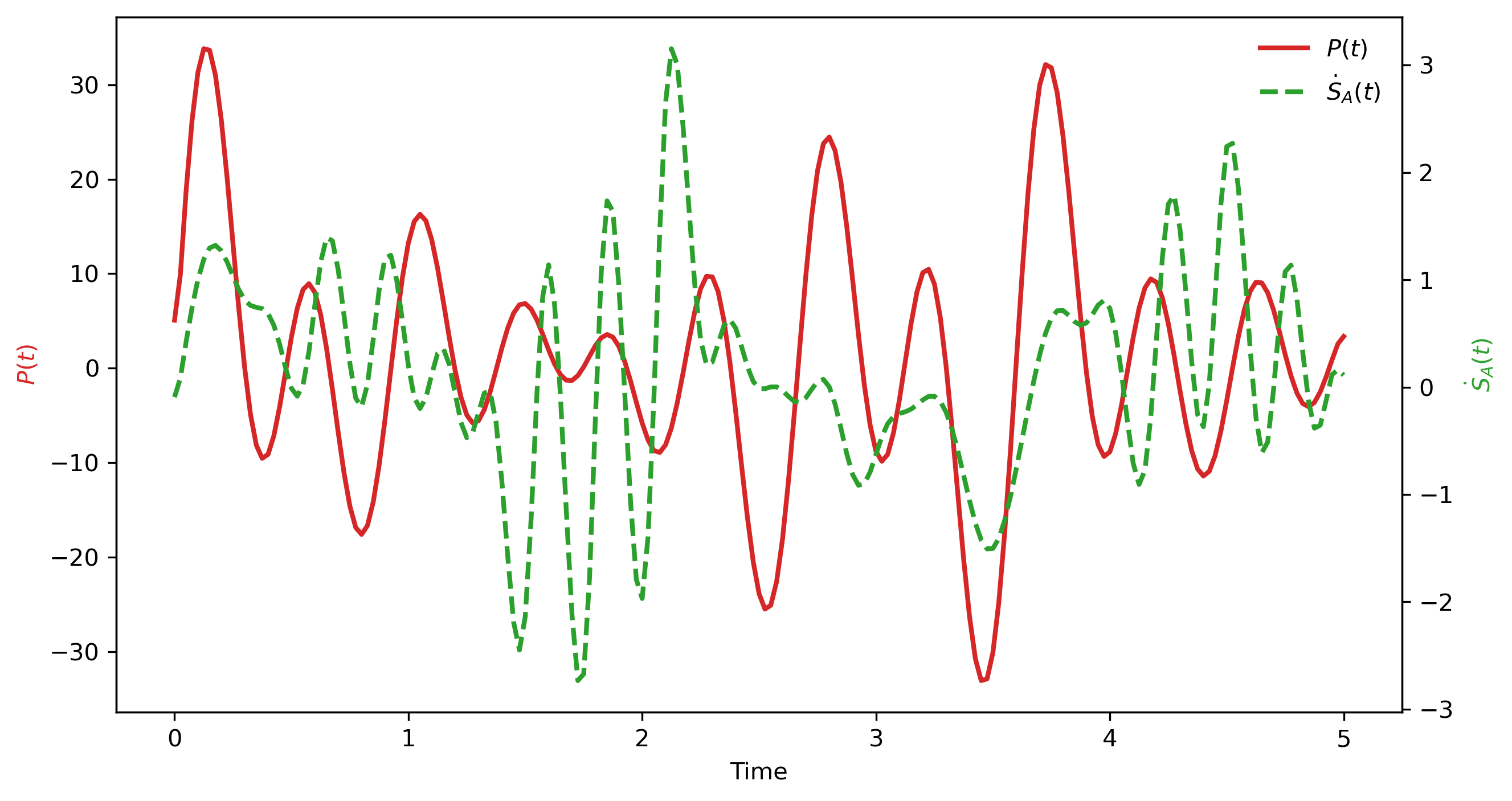}
    \hfill
 
    \vspace{2pt}
    \makebox[0.4\textwidth][c]{\textbf{(a)}}%
    \hfill
    \makebox[0.4\textwidth][c]{\textbf{(b)}}%

    \caption{(a) The bipartite entanglement entropy versus time, plotted together with the average battery energy, (b) the rate of entanglement growth versus time, plotted together with power.}
    \label{fig:entropy}
\end{figure*}
Previous works \cite{Bravyi2007EntanglingRates, VanAcoleyen2013AreaLaws} also suggests that the rate of bipartite entanglement growth depends on the strength of the interaction terms, crucially, consider a $k$-local, $g$-extensive Hamiltonian 
\begin{equation}
H=\sum_X h_X ,\qquad 
\forall i:\ \sum_{X\ni i}\|h_X\|\le g .
\end{equation}
Let $i^\star$ denote a site with maximal extensivity $g_{i^\star}=g$.  
We construct a bipartition $A|B$ by placing $i^\star\in A$ and the remaining sites in B.  
Thus the terms that cross the bipartition satisfy
\begin{equation}
\|H_{AB}\|
\;=\;
\Big\|\sum_{\substack{X\ni i^\star\\ X\ \text{crosses}\ A|B}} h_X\Big\|
\;\le\; g.
\end{equation}

For qudits of local dimension $q$,
\begin{equation}
d=\min(|A|,|B|)=q,
\qquad 
\log d= \log q .
\end{equation}
Let $\Gamma(\Psi,H)$ denote the rate of bipartite entanglement growth when the system is initially started in a pure state $\ket{\Psi}$ of the composite system $AB$. The Small Incremental Entangling (SIE) theorem \cite{Bravyi2007EntanglingRates} gives
\begin{equation}
\Gamma(\Psi,H)
\;\le\;
c\,\|H_{AB}\|\log d
\;\le\;
c g \left(\log q\right)
=
C\, g  ,
\end{equation}
where $C=c\log q=O(1)$ is an absolute constant.
Hence, for any $g$-extensive Hamiltonian and any bipartition isolating the maximally extensive site, the instantaneous entangling rate is bounded by
\begin{equation}
\Gamma(\Psi,H)\le C\, g .
\end{equation}
demonstrating the role of $g-$extensivity in bounding the rate of entanglement growth, which consequently may produce better power.

\end{document}